\documentclass[]{article}
\usepackage{proceed2e}
\usepackage{natbib}
\usepackage{url}

\usepackage{mathtools}
\usepackage{verbatim}
\usepackage{amssymb}
\usepackage{pgfplots}
\usepackage{tikz}
\usetikzlibrary{shapes,arrows,positioning}
\newcommand{\executeiffilenewer}[3]{%
\ifnum\pdfstrcmp{\pdffilemoddate{#1}}%
{\pdffilemoddate{#2}}>0%
{\immediate\write18{#3}}\fi%
}

\usepackage{pgfplots}
\pgfplotsset{compat=newest}
\pgfplotsset{plot coordinates/math parser=false}
\usepackage{graphicx}
\graphicspath{{Figures/}}
\usepackage{multirow}
\usepackage{color}
\usepackage{cancel}
\usepackage{xypic}
\usepackage{float}
\usepackage{verbatim}
\usepackage{caption}
\usepackage{subcaption}
\usepackage{amsmath,amsfonts,amsthm}
\usepackage{algorithm}
\usepackage{bm}
\usepackage{algpseudocode}
\usepackage[all]{xy}

\newtheorem{thm}{Theorem}[section]
\newtheorem{theorem}{Theorem}[section]

\newtheorem{corollary}{Corollary}

\newtheorem{lemma}{lemma}
\theoremstyle{definition}
\newtheorem{definition}{Definition}

\theoremstyle{remark}
\newtheorem{remark}{Remark}

\newcommand{\Sen}[4]{\mathrm{S}_{#2,#1}\br{{#4}}}
\newcommand{\Dcal}{\hat{\cal D}}

\newcommand{\FE}{{\rm FE}}
\newcommand{\IH}{{\rm IH}}

\newcommand{\lama}[1]{\lambda_{\max}\left({#1}\right)}

\newcommand{\opt}[1]{{#1}^\ast}

\DeclareMathOperator{\sign}{\mathrm{sign}}

\DeclareMathOperator*{\mini}{\text{Minimize}}

\DeclareMathOperator*{\argmin}{{ \sf{argmin}}}

\newcommand{\fpert}[3]{{\bar{#1}}\br{#3}}
\newcommand{\fpertr}[3]{{\tilde{#1}}_{#2}\br{{#3}}}
\newcommand{\logpaw}[4]{{#4}_{{#3}}\br{{#2}}}
\newcommand{\expb}[1]{\exp\left({#1}\right)}
\newcommand{\logb}[1]{\log\left({#1}\right)}
\newcommand{\norm}[1]{\left\lVert {#1} \right\rVert}

\newcommand{\powb}[2]{{\left({#1}\right)}^{#2}}

\newcommand{\tran}[1]{{#1}^T}

\newcommand{\inner}[2]{\left\langle {#1},{#2} \right\rangle}
\newcommand{\Rad}[1]{\mathrm{R}\br{#1}}

\DeclareMathOperator{\der}{\mathrm{d}}

\newcommand{\diff}[1]{\der{#1}}

\DeclareMathOperator*{\Es}{\mathrm{E}}

\DeclareMathOperator*{\Vars}{\mathrm{Var}}

\DeclareMathOperator{\KLs}{\mathrm{KL}}
\newcommand{\N}{\mathcal{N}}
\newcommand{\reg}{\mu}

\newcommand{\Var}[2]{\Vars_{{#1}}\left({#2}\right)}
\newcommand{\Gauss}[2]{\N\left({#1},{#2}\right)}

\newcommand{\ExP}[2]{\Es_{{#1}}\left [{#2}\right ]}

\newcommand{\KL}[2]{\KLs\left({#1 }\parallel {#2}\right)}

\newcommand{\Proj}[2]{\mathrm{Proj}_{{#1}}\left({#2}\right)}

\newcommand{\Gcal}{\mathcal{G}}
\newcommand{\Gb}{\mathbf{G}}

\newcommand{\detb}[1]{\det\left({#1}\right)}
\newcommand{\tranb}[1]{{\left({#1}\right)}^T}
\newcommand{\br}[1]{\left({#1}\right)}

\newcommand{\inv}[1]{{\left(#1\right)}^{-1}}
\newcommand{\inve}[1]{{#1}^{-1}}

\newcommand{\tr}[1]{\mathrm{tr}\left( {#1} \right)}

\newcommand{\risk}{\alpha}

\newcommand{\nf}{r}
\newcommand{\nc}{n_u}
\newcommand{\ns}{n_s}

\newcommand{\nn}{p}

\newcommand{\costu}{R}

\newcommand{\costtraj}{\mathcal{L}}

\newcommand{\costt}{\ell}
\newcommand{\cost}{\costt_t}

\newcommand{\costfin}{\costt_f}

\newcommand{\R}{\mathbf{R}}

\newcommand{\s}{x}

\newcommand{\noisetraj}{\pmb{\epsilon}}

\newcommand{\traj}{\mathbf{x}}

\newcommand{\ug}{\mathit{u}}

\newcommand{\Xspace}{\mathcal{X}}%

\newcommand{\ytraj}{\mathbf{y}}

\newcommand{\noiseu}{\omega}
\newcommand{\noise}{\omega}
\newcommand{\noisew}{\epsilon}

\DeclareMathOperator{\T}{\mathcal{T}}

\DeclareMathOperator{\Tex}{N_e}

\newcommand{\Nh}{N}

\newcommand{\Detff}[4]{\dynf\left({#1},{#2},{#3},{#4}\right)}
\newcommand{\Pn}{\mathbb{P}_{\noisew}}

\newcommand{\dynf}{\mathcal{F}}
\newcommand{\feats}{\phi}
\newcommand{\feat}[2]{\feats\left({#1},{#2}\right)}

\DeclareMathOperator{\K}{K}
\newcommand{\Kb}{\mathbf{K}}

\newcommand{\C}{\mathcal{C}}

\newcommand{\FHinf}[1]{q_\infty}
\newcommand{\FHtwo}[1]{q_2}
\newcommand{\FHone}[1]{q_1}

\newcommand{\Unoise}{\Sigma}
\newcommand{\Unoisei}{S}
\newcommand{\mup}{\overline{m}}

\title{Universal Convexification via Risk-Aversion}

\author{} 

\author{ {\bf Krishnamurthy Dvijotham} \\
Dept of Computer Science \& Engg\\
University of Washington\\
Seattle, WA 98195 \\
\And
{\bf Maryam Fazel}  \\
Dept of Electrical Engg          \\
University of Washington \\
Seattle, WA 98195
\And
{\bf Emanuel Todorov}   \\
Dept of Computer Science \& Engg \\
Dept of Applied Mathematics \\
University of Washington    \\
Seattle, WA 98195\\
}
\newcommand{\ygu}{y}
\newcommand{\opta}[1]{{#1}^{\ast}_{\risk}}

\begin{document}

\maketitle

\begin{abstract}
We develop a framework for convexifying a fairly general class of optimization problems. Under additional assumptions, we analyze the suboptimality of the solution to the convexified problem relative to the original nonconvex problem and prove additive approximation guarantees. We then develop algorithms based on stochastic gradient methods to solve the resulting optimization problems and show bounds on convergence rates. 
We then extend this framework to apply to a general class of discrete-time dynamical systems. In this context, our convexification approach falls under the well-studied paradigm of risk-sensitive Markov Decision Processes. We derive the first known model-based and model-free policy gradient optimization algorithms with guaranteed convergence to the optimal solution. Finally, we present numerical results validating our formulation in different applications.
\end{abstract}

\section{INTRODUCTION}
It has been said that the ``the great watershed in optimization isn't between linearity and nonlinearity, but convexity and nonconvexity''. In this paper, we describe a framework for convexifying a fairly general class of optimization problems (section \ref{sec:SoftCOpt}), turning them into problems that can be solved with efficient convergence guarantees. The convexification approach may change the problem drastically in some cases, which is not surprising since most nonconvex optimization problems are NP-hard and cannot be reduced to solving convex optimization problems. However, under additional assumptions, we can make guarantees that bound the suboptimality of the solution found by solving the convex surrogate relative to the optimum of the original nonconvex problem (section \ref{sec:SubOpt}). We adapt stochastic gradient methods to solve the resulting problems (section \ref{sec:Algorithms}) and prove convergence guarantees that bound the distance to optimality as a function of the number of iterations. In section \ref{sec:Control}, we extend the framework to arbitrary dynamical systems and derive the first known (to the best of our knowledge) policy optimization approach with guaranteed convergence to the global optimum. The control problems we study fall under the classical framework of risk-sensitive control and the condition required for convexity is a natural one relating the control cost, risk factor and noise covariance. It is very similar to the condition used in path integral control \citep{Fleming82,Kappen05} and subsequent iterative algorithms (the $\mathrm{PI}^2$ algorithm \citep{Theodorou2010}). In section \ref{sec:Numerical}, we present numerical results illustrating the applications of our framework to various problems.

\subsection{RELATED WORK}

Smoothing with noise is a relatively common approach that has been used extensively in various applications to simplify difficult optimization problems. It has been used in computer vision \citep{Rangarajan90generalisedgraduated} heuristically and formalized in recent work \citep{mobahi2012phd} where it is shown that under certain assumptions, adding sufficient noise eventually leads to a convex optimization problem. In this work, however, we show that for any noise level, one can choose the degree of risk-aversion in order to obtain a convex problem. In the process, one may destroy minima that are not robust to perturbations, but in several applications, it makes sense to find a robust solution, rather than the best one. Our work can be extended to control of dynamical systems (section \ref{sec:Control}). In this setting, the condition for convexity is closely related to the one that arises in path-integral control \citep{Fleming82,Kappen05}. Closely related work for  discrete-time was developed in \citep{Todorov09} and extensions and applications have been proposed in subsequent papers \citep{Theodorou2010,KappenRisk,Dvijotham2011,Toussaint09}.

\section{NOTATION}
Gaussian random variables are denoted by $\noise \sim \Gauss{\mu}{\Sigma}$, where $\mu$ is the mean and $\Sigma$ the covariance matrix. Given a random variable $\noise \in \Omega$ with distribution $P$ and a function $h: \Omega \mapsto \R$, the expected value of the random variable $h\br{\noise}$ is written as $\displaystyle\ExP{\noise \sim P}{h\br{\noise}}$. Whenever it is clear from the context, we will drop the subscript on the expected value and simply write $\ExP{}{h\br{\noise}}$. We differ from convention here slightly by denoting random variables with lowercase letters and reserve uppercase letters for matrices. We will frequently work with Gaussian perturbations of a function, so we denote:
\[\fpert{h}{\Sigma}{\theta}=\ExP{\noise \sim \Gauss{0}{\Sigma}}{h\br{\theta+\noise}},\fpertr{h}{\omega}{\theta}=h\br{\theta+\noise}-\fpert{h}{\Sigma}{\theta}.\]
$I_n$ denotes the $n \times n$ identity matrix. Unless stated otherwise, $\norm{\cdot}$ refers to the Euclidean $\ell_2$ norm. Given a convex set $\C$, we denote the projection onto $\C$ of $x$ by $\Proj{\C}{x}$: $\Proj{\C}{x}=\argmin_{y \in \C} \norm{y-x}$.
The maximum eigenvalue of a symmetric matrix $M$ is denoted by $\lama{M}$. Given symmetric matrices $A,B$, the notation $A \succeq B$ means that $A-B$ is a positive semidefinite matrix. Given a positive definite matrix $M \succ 0$, we denote the metric induced by $M$ as ${\|x\|}_M= \sqrt{\tran{x}M x}$.

\section{GENERAL OPTIMIZATION PROBLEMS}\label{sec:General}
\label{sec:SoftCOpt}
We study optimization problems of the form:
\begin{align}
\min_{\theta \in \C} g\br{\theta} \label{eq:OptGeneral}
\end{align}
where $g$ is an arbitrary function  and $\C \subset \R^k$ is a convex set. We do not assume that $g$ is convex so the above problem could be a nonconvex optimization problem. In this work, we convexify this problem by decomposing $g\br{\theta}$ as follows:
$g(\theta)=f(\theta)+\frac{1}{2}\tran{\theta}R\theta$
and perturbing $f$ with Gaussian noise. Optimization problems of this form are very common in machine learning (where $R$ corresponds to a regularizer) and control (where $R$ corresponds to a control cost). The convexified optimization problem is:
\begin{align}
\min_{\theta\in \C} \logpaw{\Sigma}{\theta}{\risk}{f}+\frac{1}{2}\tran{\theta}R\theta \label{eq:SoftConvex}
\end{align}
where $ R \succeq 0$ and \[\logpaw{\Sigma}{\theta}{\risk}{f} = \frac{1}{\risk}\logb{\ExP{\noise \sim \Gauss{0}{\Sigma}}{\expb{\risk f\br{\theta+\noise}}}}.\]
This kind of objective is common in risk-averse optimization. To a first order Taylor expansion in $\risk$, the above objective is equal to $\ExP{}{f\br{\theta+\noise}}+\risk\Var{}{f\br{\theta+\noise}}$, indicating that increasing $\risk$ will make the solution more robust to Gaussian perturbations.
$\risk$ is called the risk-factor and is a measure of the risk-aversion of the decision maker. Larger values of $\risk$ will reject solutions that are not robust to Gaussian perturbations.

In order that the expectation exists, we require that $f$ is bounded above:
\begin{align}
\text{ For all $\theta \in \R^k$, } f\br{\theta} \leq M < \infty. \label{eq:BoundedGrowth}
\end{align}
We implicitly make this assumption throughout this paper in all the stated results. Note that this is not a very restrictive assumption, since, given any function $g$ with a finite minimum, one can define a new objective $g^\prime=\min\br{g,\bar{m}}$, where $\bar{m}$ is an upper bound on the minimum (say the value of the function at some point), without changing the minimum. Since the convex quadratic is non-negative, $f$ is also bounded above by $\bar{m}$ and hence $0 < \expb{\risk f\br{\theta+\noise}}\leq \expb{\risk \mup}$. This ensures that $\logpaw{\Sigma}{\theta}{\risk}{f}$ is finite. Some results will require differentiability, and  we can preserve this by defining $g^\prime$ using a soft-min: For example,  $g^\prime\br{x}=\bar{m}\tanh\br{\frac{g\br{x}}{\bar{m}}}$.

\begin{theorem}[Universal Convexification]\label{thm:ConvexityMain}
The optimization problem \eqref{eq:SoftConvex} is a convex optimization problem whenever $\C$ is a convex set and $\risk R \succeq \inve{\Sigma}$.
\end{theorem}
\begin{proof}
Writing out the objective function in \eqref{eq:SoftConvex} and scaling by $\risk$, we get:
\begin{align}
\logb{\ExP{\noise \sim \Gauss{0}{\Sigma}}{\expb{\risk f(\theta+\noise)}}}+\frac{\risk}{2}\tran{\theta}R\theta\label{eq:Thm1s1}
\end{align}
Writing out the expectation, we have:
\begin{align*}
& \ExP{\noise \sim \Gauss{0}{\Sigma}}{\expb{\risk f(\theta+\noise)}} = \ExP{\ygu \sim \Gauss{\theta}{\Sigma}}{\expb{\risk f(\ygu)}} \\
& \propto \int \expb{-\frac{1}{2}\tranb{\ygu-\theta}\inve{\Sigma}\br{\ygu-\theta}}\expb{\risk f(\ygu)}\diff{\ygu}
\end{align*}
where we omitted the normalizing constant ${\br{\sqrt{\powb{2\pi}{n}\detb{\Sigma}}}}^{-1}$ in the last step. Thus, exponentiating \eqref{eq:Thm1s1}, we get (omitting the normalizing constant):
\[\int \expb{\frac{\risk}{2} \tran{\theta}R\theta-\frac{1}{2}\tranb{\ygu-\theta}\inve{\Sigma}\br{\ygu-\theta}+\risk f(\ygu)}\diff{\ygu}\]
The term inside the exponent can be rewritten as
\[\frac{1}{2}\tran{\theta}\br{\risk R-\inve{\Sigma}}\theta+\tran{\theta}\inve{\Sigma}\ygu+\risk f\br{\ygu}-\frac{1}{2}\tran{\ygu}\inve{\Sigma}\ygu\]
Since $\alpha R\succeq \inve{\Sigma}$, this is a convex quadratic function of $\theta$ for each $\ygu$. Thus, the overall objective is the composition of a convex and increasing function $\log{\ExP{}{\expb{\cdot}}}$ and a convex quadratic and is hence convex \citep{Boyd04}. Since $\C$ is convex, the overall problem is a convex optimization problem.
\end{proof}

\subsection{INTERPRETATION}

Theorem \ref{thm:ConvexityMain} is a surprising result, since the condition for convexity does not depend in any way on the properties of $f$ (except for the bounded growth assumption \eqref{eq:BoundedGrowth}), but only on the relationship between the quadratic objective $R$, the risk factor $\risk$ and the noise level $\Sigma$. In this section, we give some intuition behind the result and describe why it is plausible that this is true.

In general, arbitrary nonconvex optimization problems can be very challenging to solve. As a worst case example, consider a convex quadratic function $g(x)=x^2$ that is perturbed slightly: At some point where the function value is very large (say $x=100$), we modify the function so that it suddenly drops to a large negative value (lower than the global minimum 0 of the convex quadratic). By doing this perturbation over a small finite interval , one can preserve differentiability while introducing a new global minimum far away from the original global minimum. In this way, one can create difficult optimization problems that cannot be solved using gradient descent  methods, unless initialized very carefully.

The work we present here does not solve this problem: In fact, it will not find a global minimum of the form created above. The risk-aversion introduced destroys this minimum, since small perturbations cause the function to increase rapidly, ie, $g\br{\opt{\theta}+\noiseu}\gg g\br{\opt{\theta}}$, so that the objective \eqref{eq:SoftConvex} becomes large. Instead, it will find a ``robust'' minimum, in the sense that Gaussian perturbations around the minimum do not increase the value of the objective by much. This intuition is formalized by theorem \ref{thm:GenSubOpt}, which bounds the suboptimality of the convexified solution relative to the optimal solution of the original problem in terms of the sensitivity of $f$ to Gaussian perturbations around the optimum.

In figure \ref{fig:ConvexEffect}, we illustrate the effect of the convexification for a 1-dimensional optimization problem. The blue curve represents the original function $g\br{\theta}$. It has 4-local minima in the interval $(-3,3)$. Two of the shallow minima are eliminated by smoothing with Gaussian noise to get $\fpert{g}{\Sigma}{\theta}$. However, there is a deep but narrow local minimum that remains even after smoothing. Making the problem convex using risk-aversion and theorem \ref{thm:ConvexityMain} leads to the green curve that only preserves the robust minimum as the unique global optimum.

\begin{figure}[htb]
\begin{center}
\includegraphics[width=.6\columnwidth]{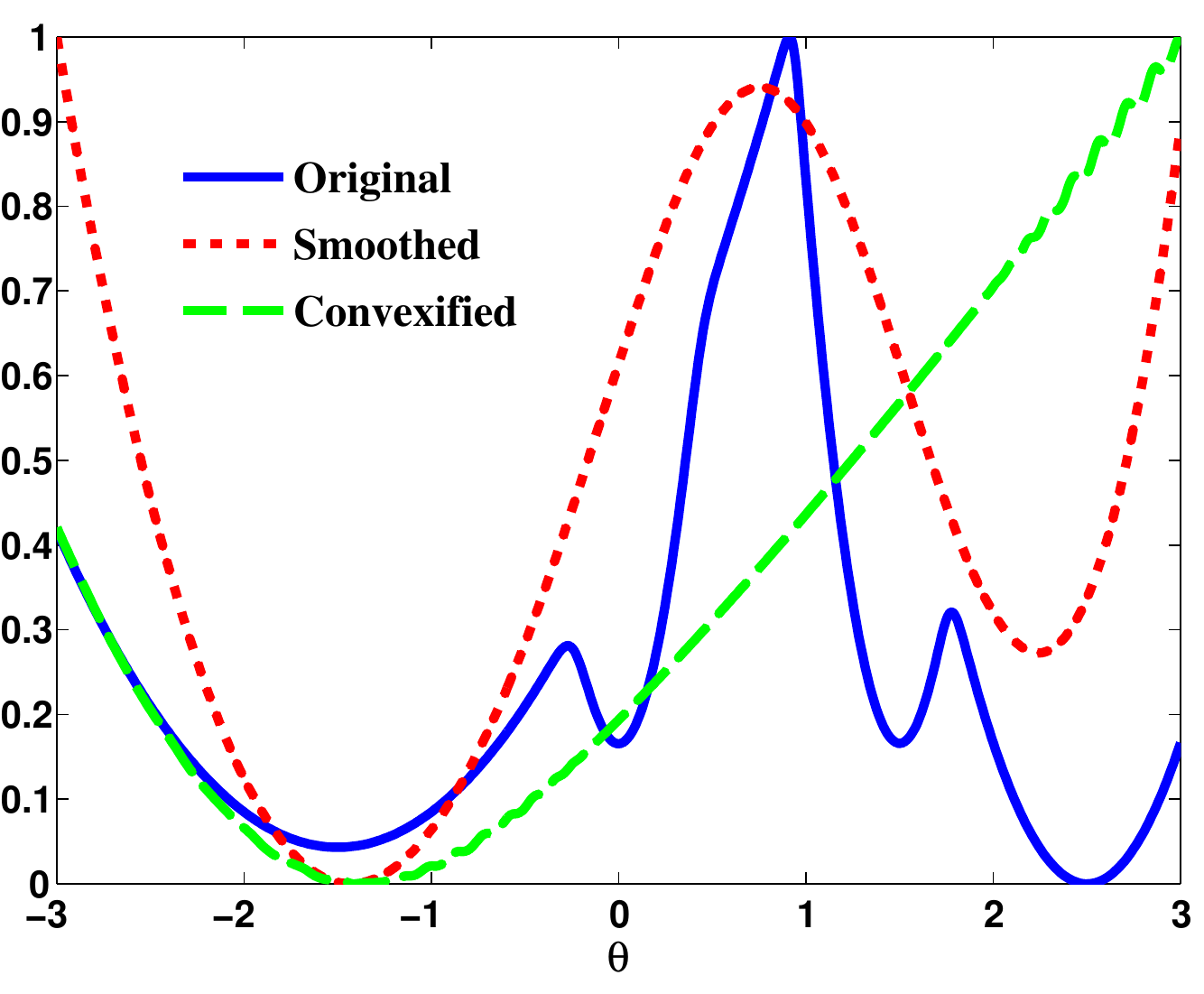}
\end{center}
\caption{Illustration of Convexification for a 1-dimensional optimization problem}\label{fig:ConvexEffect}
\end{figure}


\subsection{ANALYSIS OF SUB-OPTIMALITY}\label{sec:SubOpt}

We have derived a convex surrogate for a very general class of optimization problems. However, it is possible that the solution to the convexified problem is drastically different from that of the original problem and the convex surrogate we propose is a poor approximation. Given the hardness of general non-convex optimization, we do not expect the two problems to have close solutions without additional assumptions. In this section, we analyze the gap between the original and convexified problem, that is, we answer the question: Can we construct a reasonable approximation to the original (potentially nonconvex) problem based on the solution to the perturbed convex problem? In order to answer this, we first define the sensitivity function, which quantifies the gap between \eqref{eq:OptGeneral} and \eqref{eq:SoftConvex}.

\begin{definition}[Sensitivity Function]
The sensitivity function of $f$ at noise level $\Sigma$, risk $\risk$ is defined as:
\begin{align*}
\Sen{f}{\risk}{\Sigma}{\theta} & =\frac{1}{\risk}\logb{\displaystyle\ExP{}{\expb{\risk\br{f\br{\theta+\noise}-\fpert{f}{\Sigma}{\theta}}}}} \\
                               & = \frac{1}{\risk}\logb{\ExP{}{\expb{\risk \fpertr{f}{\noise}{\theta}}}}.
\end{align*}
This is a measure of how sensitive $f$ is to Gaussian perturbations  at $\theta$.

\begin{theorem}[Suboptimality Analysis]\label{thm:GenSubOpt}
Let $g(\theta)=f(\theta)+\frac{1}{2}\tran{\theta}R\theta$ and define:
\[\opta{\theta}=\argmin_{\theta \in \C} \logpaw{\Sigma}{\theta}{\risk}{f}+\frac{1}{2}\tran{\theta}R\theta, \quad \opt{\theta}=\argmin_{\theta \in \C}\fpert{g}{\Sigma}{\theta}.\]
Then,
\begin{align*}
&\fpert{g}{\Sigma}{\opta{\theta}}-\fpert{g}{\Sigma}{\opt{\theta}}\leq \Sen{f}{\risk}{\Sigma}{\theta}
\end{align*}
\begin{proof}
We have that $\forall \theta$,
$\fpert{g}{\Sigma}{\theta}= \fpert{f}{\Sigma}{\theta}+\frac{1}{2}\tran{\theta}R\theta+\frac{1}{2}\tr{\Sigma R}$. From the convexity of the function $y \to \expb{\alpha y}$ and Jensen's inequality, we have
\begin{align*}
& \fpert{g}{\Sigma}{\opta{\theta}}  = \fpert{f}{\Sigma}{\opta{\theta}}+\frac{1}{2}\tran{\opta{\theta}}R\opta{\theta}+\frac{1}{2}\tr{\Sigma R} \\
& \leq \logpaw{\Sigma}{\opta{\theta}}{\risk}{f}+\frac{1}{2}\tran{\opta{\theta}}R\opta{\theta}+\frac{1}{2}\tr{\Sigma R} 
\end{align*}
Since $\opta{\theta}=\argmin_{\theta \in \C} \logpaw{\Sigma}{\theta}{\risk}{f}+\frac{1}{2}\tran{\theta}R\theta$, for any $\theta \in \C$, we have
\begin{align*}
\fpert{g}{\Sigma}{\opta{\theta}} &\leq \logpaw{\Sigma}{\theta}{\risk}{f}+\frac{1}{2}\tran{\theta}R\theta+\frac{1}{2}\tr{\Sigma R} \\
& = \Sen{f}{\risk}{\Sigma}{\theta} +\fpert{g}{\Sigma}{\theta}.
\end{align*}
 The result follows by plugging in $\theta=\opt{\theta}$ and subtracting $\fpert{g}{\Sigma}{\opt{\theta}}$ from both sides.
\end{proof}

\end{theorem}

\begin{remark}
Although we only prove suboptimality relative to the optimal solution of a smoothed version of $g$, we can extend the analysis to the optimal solution of $g$ itself. Define $\opt{\theta}=\argmin_{\theta \in \C} g\br{\theta}$. We can prove that
\begin{align*}
g\br{\opta{\theta}}-g\br{\opt{\theta}} &\leq \br{\fpert{g}{\Sigma}{\opt{\theta}}-g\br{\opt{\theta}}}+\br{g\br{\opta{\theta}}-\fpert{g}{\Sigma}{\opta{\theta}}} \\
& \quad + \Sen{f}{\risk}{\Sigma}{\theta}
\end{align*}
Assuming that $g$ changes slowly around $\opt{\theta}$ and $\opta{\theta}$ (indicative of the fact that $\opt{\theta}$ is a ``robust'' minimum and $\opta{\theta}$ is the minimum of a robustified problem), we can bound the first term. We leave a precise analysis for future work.
\end{remark}



\end{definition}

\subsubsection{BOUNDING THE SENSITIVITY FUNCTION}
The sensitivity function is exactly the moment generating function of the $0$-mean random variable $\fpertr{f}{\noise}{\theta}$.
Several techniques have been developed for bounding moment generation functions in the field of concentration inequalities \citep{boucheron2013concentration}. Using these techniques, we can bound the moment generating function (i.e. the sensitivity function) under the assumption that $f$ is Lipschitz-continuous. Before stating the lemma, we state a classical result:
\begin{theorem}[Log-Sobolev Inequality, \citep{boucheron2013concentration}, theorem 5.4]\label{thm:LogSobolev}
Let $\omega \sim \Gauss{0}{\sigma^2 I}$ and $f$ be any continuously differentiable function of $\omega$. Then,
\[\ExP{}{{f\br{\omega}}^{2}\logb{\frac{{f\br{\omega}}^{2}}{\ExP{}{f\br{\omega}^{2}}}}} \leq 2 \sigma^2   \ExP{}{\norm{\nabla f\br{\omega}}^2}.\]
Further, if $f$ is Lipschitz with Lipschitz constant $L$, then
\[\ExP{}{\expb{\risk\br{f\br{\noise}-\ExP{}{f\br{\noise}}}}} \leq \expb{\risk^2 L^2\sigma^2/2}.\]
Finally, we can also prove that if f is differentiable and $\ExP{}{\exp{\lambda \norm{\nabla f\br{\noise}}^2 }}<\infty$ for all $\lambda<\lambda_0$, then, given $\eta$ such that $\lambda\sigma^2<\eta\lambda_0$ and $\lambda \eta <2$, we have
\begin{align*}
&\logb{\ExP{}{\expb{\lambda\br{f\br{\noise}-\ExP{}{f\br{\noise}}}}}} \leq \\
&\frac{\lambda \eta}{2-\lambda \eta} \logb{\ExP{}{\expb{\frac{\lambda\sigma^2}{\eta} \norm{\nabla f\br{\noise}}^2}}}
\end{align*}
\end{theorem}
\begin{theorem}[Lipschitz Continuous Functions]\label{thm:LipSubOpt}
Assume the same setup as theorem \ref{thm:GenSubOpt}. Suppose further that $f$ is Lipschitz continuous with Lipschitz constant $L$, that is,
\[\forall \theta,\theta^\prime, |f(\theta)-f(\theta^\prime)|\leq L\norm{\theta-\theta^\prime} .\]
Then, $\Sen{f}{\risk}{\Sigma}{\theta}\leq \frac{1}{2}\risk L^2 \lama{\Sigma}$. Further,
\[\fpert{g}{\Sigma}{\theta}-\min_{\theta \in \C}\fpert{g}{\Sigma}{\theta}\leq \frac{\risk L^2 \lama{\Sigma}}{2}. \]
\end{theorem}
\begin{proof}
We can write \[\tilde{f}\br{\noise^\prime}=f\br{\theta+\Sigma^{1/2}\noise^\prime}-\fpert{f}{\Sigma}{\theta}\]
 where $\noise^\prime \sim \Gauss{0}{I_k}$.
 \begin{align*}
 &|\tilde{f}\br{\omega^\prime_1}-\tilde{f}\br{\omega^\prime_2}|=|f\br{\theta+\Sigma^{1/2}\noise^\prime_1}-f\br{\theta+\Sigma^{1/2}\noise^\prime_2}| \\
 &\leq L \norm{\Sigma^{1/2}\br{\omega^\prime_1-\omega^\prime_2}} \leq L \lama{\Sigma^{1/2}}\norm{\omega^\prime_1-\omega^\prime_2}
\end{align*}
Since $\lama{\Sigma^{1/2}}=\sqrt{\lama{\Sigma}}$, this shows that $\tilde{f}$ is Lipschitz with Lipschitz constant $\sqrt{\lama{\Sigma}}L$. The result now follows from theorem \ref{thm:LogSobolev}. From theorem \ref{thm:GenSubOpt}, this implies that
 \[\fpert{g}{\Sigma}{\theta}-\min_{\theta \in \C}\fpert{g}{\Sigma}{\theta}\leq \frac{\risk L^2 \lama{\Sigma}}{2}. \]

\end{proof}

\section{ALGORITHMS AND CONVERGENCE GUARANTEES}\label{sec:Algorithms}

In general, the expectations involved in \eqref{eq:SoftConvex} cannot be computed analytically. Thus, we need to resort to sampling based approaches in order to solve these problems. This has been studied extensively in recent years in the context of machine learning, where stochastic gradient methods and variants have been shown to be efficient, particularly in the context training machine learning algorithms with huge amounts of data. We now describe stochastic gradient methods for solving \eqref{eq:SoftConvex} and adapt the convergence guarantees of stochastic gradient methods to our setting.


\subsection{Stochastic Gradient Methods with Convergence Guarantees}\label{sec:GradUpdate}
In this section, we will derive gradients of the convex objective function  \eqref{eq:SoftConvex}.  We will assume that the function $f$ is differentiable at all $\theta \in \R^k$. In order to get unbiased gradient estimates, we exponentiate the objective \eqref{eq:SoftConvex} to get:
\[\Gb\br{\theta}=\ExP{\noise \sim \Gauss{0}{\Sigma}}{\expb{\risk f\br{\theta+\noise}+\frac{\risk}{2}\tran{\theta}R\theta}}.\]
Since $f\br{\theta}$ is differentiable for all $\theta$, so is $\expb{\risk f\br{\theta+\noise}+\frac{1}{2}\tran{\theta}\br{\risk R}\theta}$. Further, suppose that
\[\ExP{}{\expb{2\risk f\br{\theta+\noise}}\norm{ \nabla f\br{\theta+\noise} + R \theta}^2}\]
exists and is finite for each $\theta$. Then, if we differentiate $\Gb\br{\theta}$ with respect to $\theta$, we can interchange the expectation and differentiation to get
\begin{align*}
\ExP{}{\expb{\risk f\br{\theta+\noise}+\frac{\risk}{2}\tran{\theta}R\theta}\br{\risk \nabla f\br{\theta+\noise}+\risk R \theta}}
\end{align*}
Thus, we can sample $\noise \sim \Gauss{0}{\Sigma}$ and get an unbiased estimate of the gradient
\begin{align}
\risk \expb{\risk f\br{\theta+\noise}+\frac{\risk}{2}\tran{\theta}R\theta}\br{\nabla f\br{\theta+\noise}+ R \theta}\label{eq:UnbiasedGrad1}
\end{align}
which we denote by $\hat{\nabla}\Gb\br{\theta,\noise}$. As in standard stochastic gradient methods, one saves on the complexity of a single iteration by using a single (or a small number of) samples to get a gradient estimate while still converging to the global optimum with high probability and in expectation, because over multiple iterations one moves along the negative gradient ``on average''.
\begin{algorithm}[Stochastic Gradient Method]
 \begin{algorithmic}
 \State $\theta \gets 0$
\For {$i = 1,\ldots,T$}
    \State $\noise \sim \Gauss{0}{\Sigma},\theta \gets  \Proj{\C}{\theta-\eta_i\hat{\nabla}\Gb\br{\theta,\noise}}$
\EndFor
\end{algorithmic}
 \caption{Stochastic Gradient Method for \eqref{eq:SoftConvex}}\label{alg:StochGrad}
\end{algorithm}

From standard convergence theory for stochastic gradient methods \citep{BubeckNotes}, we have:
\begin{corollary}
Suppose that $\ExP{}{\norm{\hat{\nabla}\Gb\br{\theta,\noise}}^2}\leq \zeta^2$, \\$\C$ is contained in a ball of radius $\Rad{\C}$ and \\$ \risk R \succeq \inve{\Sigma}$. Run algorithm \ref{alg:StochGrad} for $T$ iterations with \\$\eta_i=\frac{\Rad{\C}}{\zeta}\sqrt{\frac{1}{2i}}$ and define $\hat{\theta}=\frac{1}{T}\sum_{i=1}^T \theta_i$, $\opt{\Gb}=\argmin_{\theta\in \C}\Gb\br{\theta}$. Then, we have
\[\ExP{}{\Gb\br{\hat{\theta}}}-\opt{\Gb}\leq \Rad{\C}\zeta \sqrt{\frac{1}{2T}}\]
\end{corollary}

In the following theorem, we prove a convergence rate for algorithm \ref{alg:StochGrad}.
\begin{theorem}\label{thm:AlgConv}
Suppose that $R=\kappa I, \Sigma =\sigma^2 I,\risk \kappa \succeq \frac{1}{\sigma^2}$, $\C$ is contained in a sphere of radius $\Rad{\C}$ and that for all $\theta \in \C$ $\fpert{g}{\Sigma}{\theta} \leq \mup$. Also, suppose that for some $\beta<\powb{\risk}{-1}$:
\[\frac{1}{2\risk}\logb{\ExP{\noise\sim\Gauss{0}{\Sigma}}{\expb{2\frac{\risk\sigma^2}{\beta}\norm{\nabla f\br{\theta+\noise}}^2}}}\leq \gamma^2\]
Define $\delta=\sqrt{\frac{\beta \gamma^2 }{\sigma^2\br{1-\risk\beta}}}+\kappa\Rad{\C}$. Then, we can choose
$\zeta \leq \risk^2\delta^2\expb{2\risk(\mup+\gamma^2)+\frac{\risk\beta}{1-\risk\beta}-\sigma^2\kappa}$.
\end{theorem}
\begin{remark}
The convergence guarantees are in terms of the exponentiated objective $\Gb\br{\theta}$. We can convert these into bounds on $\logb{\Gb\br{\theta}}$ as follows:
\begin{align*}
& \ExP{}{\logb{\Gb\br{\hat{\theta}}}} \leq \logb{\ExP{}{\Gb\br{\hat{\theta}}}}  \leq \logb{\opt{\Gb}+\frac{\zeta\Rad{\C}}{\sqrt{2T}}}
\end{align*}
where the first inequality follows from concavity of the $\log$ function. Subtracting $\logb{\opt{\Gb}}$, we get
\begin{align*}
&\ExP{}{\logb{\Gb\br{\hat{\theta}}}}-\logb{\opt{\Gb}} \leq \logb{1+\frac{\zeta\Rad{\C}}{\sqrt{2T}\opt{\Gb}}}.
\end{align*}
\end{remark}
\section{CONTROL PROBLEMS}\label{sec:Control}
In this section, we extend the above approach to the control of discrete-time dynamical systems. Stochastic optimal control of nonlinear systems in general is a hard problem and the only known general approach is based on dynamic programming, which scales exponentially with the dimension of the state space. Algorithms that approximate the solution of the dynamic program directly (approximate dynamic programming) have been successful in various domains, but scaling these approaches to high dimensional continuous state control problems has been challenging. In this section, we pursue the alternate approach of policy search or policy gradient methods \citep{Baxter2001}. These algorithms have the advantage that they are directly optimizing the performance of a control policy as opposed to a surrogate measure like the error in the solution to the Bellman equation. They have been used successfully for various applications and are closely related to path integral control \citep{Kappen05,theodorou2010generalized}. However, in all of these approaches, there were no guarantees made regarding the optimality of the policy that the algorithm converges to (even in the limit of infinite sampling) or the rate of convergence.

In this work, we develop the \emph{first} policy gradient algorithms that achieve the \emph{globally} optimal solutions to a class of \emph{risk-averse} policy optimization problems. 

\subsection{Problem Setup}

We deal with arbitrary discrete-time dynamical systems of the form
\begin{align}
& \noisetraj = \tran{\begin{pmatrix} \noisew_1 & \ldots &\noisew_{\Nh-1}\end{pmatrix}} \sim \Pn \nonumber \\
& \s_{1} = 0,\quad\s_{t+1}=\Detff{\s_t}{\ygu_t}{\noisew_t}{t} \quad t=1,\ldots,\Nh-1 \label{eq:Dyn}\\
& \ygu_t = \ug_t + \noiseu_t,\quad\noiseu_t \sim \Gauss{0}{\Unoise_t}  \quad t=1,\ldots,\Nh-1\label{eq:UNoise}
\end{align}
where $\s_t \in \R^{\ns}$ denotes the state, $\ygu_t \in \R^{\nc}$ the effective control input, $\noisew_t \in \R^{\nn}$ external disturbances, $\ug_t \in \R^{\nc}$ the actual control input, $\noiseu_t \in \R^{\nc}$ the control noise, $\dynf: \R^{\ns} \times \R^{\nc} \times \R^{\nn} \times \{1,\ldots,\Nh-1\}\mapsto \R^{\ns}$ the discrete-time dynamics. In this section, we will use boldface to denote quantities stacked over time (like $\noisetraj$).
Equation \eqref{eq:Dyn} can model any noisy discrete-time dynamical system, since $\dynf$ can be any function of the current state, control input and external disturbance (noise). However, we require that all the control dimensions are affected by Gaussian noise as in \eqref{eq:UNoise}. This can be thought of either as real actuator noise or artificial exploration noise. The choice of zero initial state $\s_1=0$ is arbitrary - our results even extend to an arbitrary distribution over the initial state.

We will work with costs that are a combination of arbitrary state costs and quadratic control costs:
\begin{align}
\sum_{t=1}^{\Nh} \cost(\s_t) + \sum_{t=0}^{\Nh-1} \frac{\tran{\ug_t}\costu_t\ug_t}{2} \label{eq:Obj}
\end{align}
where $\cost\br{\s_t}$ is the stage-wise state cost at time $t$. $\cost$ can be any bounded function of the state-vector $\s_t$. Further, we will assume that the control-noise is non-degenerate, that is $\Unoise_t$ is full rank for all $0 \leq t \leq \Nh-1$. We denote $\Unoisei_t = \inve{\Unoise_{t}}$. We seek to design feedback policies
\begin{align}
u_t &=K_t\feat{\s_t}{t}, \feats: \R^{\ns} \times \{1,2,\ldots,\Nh-1\} \mapsto \R^{\nf} \nonumber \\
 & \quad K_t \in \R^{\nc \times \nf} \label{eq:Pol}
\end{align}
to minimize the accumulated cost \eqref{eq:Obj}. We will assume that the features $\feats$ are fixed and we seek to optimize the policy parameters $\Kb=\{K_t:t=1,2,\ldots,\Nh-1\}$.
The stochastic optimal control problem we consider is defined as follows:
\begin{align}
 \mini_{\Kb} & \ExP{\noisetraj,\noiseu_t}{\expb{\risk\costtraj\br{\Kb}}} \nonumber\\
 \text{Subject to } &  \s_{1} = 0,\s_{t+1}=\Detff{\s_t}{\ygu_t}{\noisew_t}{t} \nonumber\\
& \ygu_t = \ug_t + \noiseu_t,\ug_t=\K_t\feat{\s_t}{t} \nonumber \\
& \noisetraj \sim \Pn,\noiseu_t \sim \Gauss{0}{\Unoise_t} \nonumber \\
& \costtraj\br{\Kb}=\sum_{t=0}^{\Nh} \cost\br{\s_t} + \sum_{t=0}^{\Nh-1} \frac{\tran{\ug_t}\costu_t\ug_t}{2}\label{eq:Opt}
\end{align}
This is exactly the same as the formulation in Risk Sensitive Markov Decision Processes \citep{marcus1997risk}, the only change being that we have explicitly separated the noise appearing in the controls from the noise in the dynamical system overall. In this formulation, the objective depends not only on the average behavior of the control policy but also on variance and higher moments (the tails of the distribution of costs). This has been studied for linear systems under the name of LEQG control \citep{speyer1974optimization}. $\risk$ is called the risk factor: Large positive values of $\risk$ result in  strongly risk-averse policies while large negative values result in risk-seeking policies. In our formulation, we will need a certain minimum degree of risk-aversion for the resulting policy optimization problem to be convex.

\subsection{Convex Controller Synthesis} \label{sec:MainControl}

\begin{thm} \label{thm:ConvObj}
If $\risk \costu_t \succeq \inv{\Unoise_t}=\Unoisei_t$ for $t=1,\ldots,\Nh-1$, then the optimization problem \eqref{eq:Opt} is convex.
\end{thm}
\begin{proof}
We first show that for a fixed $\noisetraj$, the quantity $\ExP{\noiseu_t \sim \Gauss{0}{\Unoise_t}}{\expb{\risk\costtraj\br{\Kb}}}$
is a convex function of $\Kb$. Then, by the linearity of expectation, so is the original objective. We can write down the above expectation (omitting the normalizing constant of the Gaussian) as:
\begin{align*}
\int \expb{-\displaystyle\sum_{t=1}^{\Nh-1} \frac{1}{2}\norm{\ygu_t-K_t\feat{\s_t}{t}}_{{\Unoisei_t}}^{2}+\risk\costtraj\br{\Kb}}d\ytraj
\end{align*}
If we fix $\ytraj,\noisetraj$, using \eqref{eq:Dyn}, we can construct $\s_t$ for every $t=1,\ldots,\Nh$. Thus, $\traj$ is a deterministic function of $\ytraj,\noisetraj$ and does not depend on $\Kb$. The term inside the exponential can be written as
\begin{align*}
&-\frac{1}{2}\br{\sum_{t=1}^{\Nh-1} \norm{\ygu_t}^2_{S_t}}+\risk\br{\sum_{t=1}^{\Nh}\cost(\s_t)}\\
&\quad +\sum_{t=1}^{\Nh-1} \frac{1}{2}\tr{\br{\tran{K_t}\br{\alpha\costu_t-S_t}K_t}\feat{\s_t}{t}\tran{\feat{\s_t}{t}}}\\
& \quad -\sum_{t=1}^{\Nh-1} \tran{\ygu_t}S_t\feat{\s_t}{t}K_t
\end{align*}
The terms on the first line don't depend on $\Kb$. The function $\br{\tran{K_t}\br{\alpha\costu_t-S_t}K_t}$ is convex in $\Kb$ with respect to the semidefinite cone \citep{Boyd04} when $\alpha\costu_t-S_t \succeq 0$ and $\feat{\s_t}{t}\tran{\feat{\s_t}{t}}$ is a positive semidefinite matrix. Hence the term on the second line is convex in $\Kb$. The term on the third line is linear in $\Kb$ and hence convex. Since $\exp$ is a convex and increasing function, the composed function (which is the integrand) is convex as well in $\Kb$. Thus, the integral is convex in $\Kb$.
\end{proof}

We can add arbitrary convex constraints and penalties on $\Kb$ without affecting convexity.

\begin{corollary}
The problem
\begin{align}
& \min_{\Kb}  & \ExP{\noisetraj \sim \Pn,\noiseu_t \sim \Gauss{0}{\Unoise_t}}{\expb{\risk\costtraj\br{\Kb}}} \nonumber\\
& \text{Subject to} &\eqref{eq:Dyn},\eqref{eq:UNoise},\Kb \in \C \label{eq:OptC}
\end{align}
is a convex optimization problem for any arbitrary convex set $\C \subset \R^{\nc \times \nf \times (\Nh-1)}$ if $\risk \costu_t \succeq \Sigma_t \quad \forall t$.
\end{corollary}

\section{NUMERICAL RESULTS}\label{sec:Numerical}
In this section, we present preliminary numerical results illustrating applications of the framework to various problems with comparisons to a simple baseline approach. These are not meant to be thorough numerical comparisons but simple illustrations of the power and applications of our framework.

\subsection{BINARY CLASSIFICATION}\label{sec:NumSupervised}
We look at a problem of binary classification. Let $y$ denote the actual label and $\hat{y}$ denote the predicted label. We use  the loss function
\[\costt\br{y,\hat{y}}=\begin{cases} -\infty \text{ if } y\hat{y}>0 \\ 0 \text{ otherwise }\end{cases}. \]
This is a non-convex loss function (the logarithm of the standard $0$-$1$ loss). We convexify this in the prediction $\hat{y}$ using our approach:
\[\frac{1}{\risk}\logb{\ExP{}{\expb{\risk \costt\br{y,\hat{y}}}}}+\frac{\powb{\hat{y}}{2}}{2\sigma^2}.\]
Plugging in $\hat{y}=\tran{\theta}\s$ where $\s$ is the feature vector, we get
\[\frac{1}{\risk}\logb{\ExP{\noise \sim \Gauss{0}{\sigma^2}}{\expb{\risk \costt\br{y,\tran{\theta}\s+\noise}}}}+\frac{\powb{\tran{\theta}{\s}}{2}}{2\risk\sigma^2}.\]
Plugging in the expression for $\costt$ gives
\[\frac{1}{\risk}\logb{\frac{1}{2}\mathrm{erfc}\br{\frac{y\tran{\theta}\s}{\sqrt{2}\sigma}}}+\frac{\powb{\tran{\theta}{\s}}{2}}{2\risk\sigma^2}.\]
where $\mathrm{erfc}$ is the Gaussian error function. Given a dataset $\{\br{x_i,y_u}\}$, we can form the empirical risk-minimization problem with this convexified objective:
\begin{align*}
\sum_{i=1}^M \frac{1}{\risk}\logb{\frac{1}{2}\mathrm{erfc}\br{\frac{y^i\tran{\theta}\s^i}{\sqrt{2}\sigma}}}+\frac{\powb{\tran{\theta}{\s^i}}{2}}{2\risk\sigma^2}
\end{align*}
We can drop the $\risk$ since it only scales the objective (this is a consequence of the fact that $\expb{\costt}$ is 0-1 valued and does not change on raising it to a positive power). Thus, we finally end up with
\begin{align*}
\frac{1}{M}\sum_{i=1}^M \br{\logb{\frac{1}{2}\mathrm{erfc}\br{\frac{y^i\tran{\theta}\s^i}{\sqrt{2}\sigma}}}+\frac{\powb{\tran{\theta}{\s^i}}{2}}{2\sigma^2}}.
\end{align*}
The first term is a data-fit term (a smoothed version of the 0-1 loss) and the second term is a regularizer. although we penalize the prediction $\tran{\theta}{\s}$ rather than $\theta$ itself. If $\s$ are normalized and span all directions, by summing over the entire dataset we get something close to the standard Tikhonov regularization.

We compare the performance of our convexification-based approach with a standard implementation of a Support Vector Machine (SVM) \citep{CC01a}. We use the breast cancer dataset from \citep{UCIDataset}. We compare the two algorithms on various train-test splits of the dataset (without using cross-validation or parameter tuning). For each split, we create a noisy version of the dataset by adding Gaussian noise to the labels and truncating to $+1/-1$:$\hat{y}^i=\sign\br{y^i+\noise}, \noise \sim \Gauss{0}{\sigma^2}$. The accuracy of the learned classifiers  on withheld test-data, averaged over $50$ random train-test splits with label corruption as described above, are plotted as function of the noise level $\sigma$ in figure \ref{fig:ConvexEffect}. This is not a completely fair comparison since our approach explicitly optimizes for the worst case under Gaussian perturbations to the prediction (which can also be seen as a Gaussian perturbation to the label). However, as mentioned earlier, the purpose of these numerical experiments is to illustrate the applicability of our convexification approach to various problems so we do not do a careful comparison to robust variants of SVMs, which would be better suited to the setting described here.
\begin{figure}[htb]
\begin{center}
\includegraphics[width=.8\columnwidth]{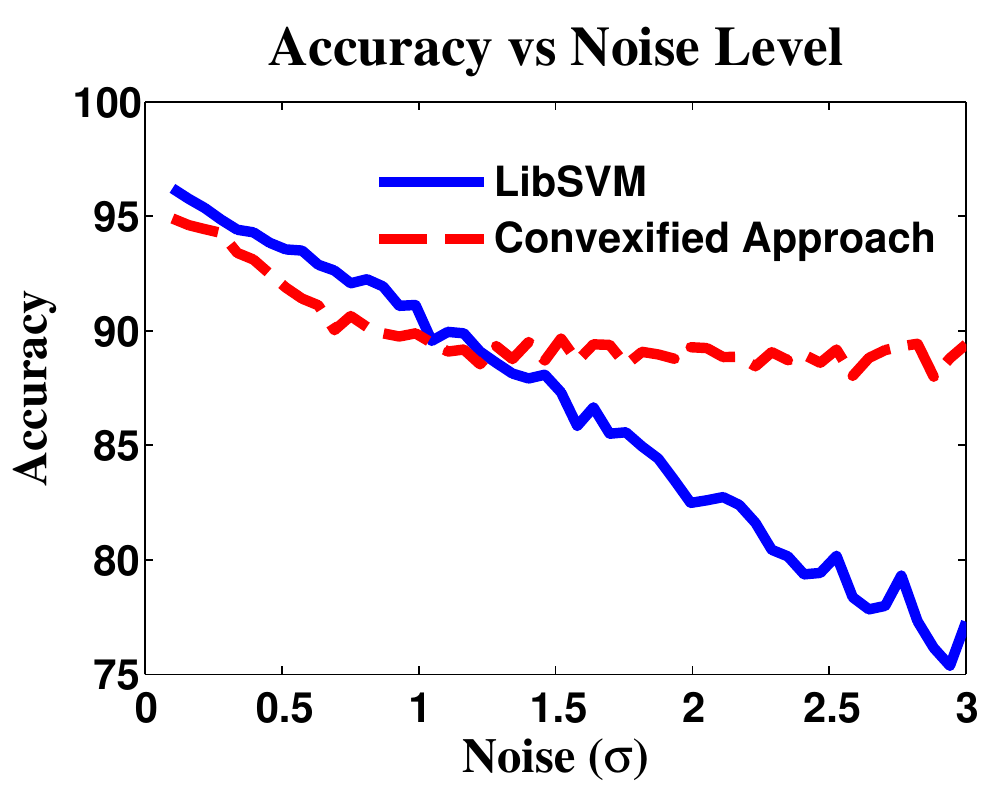}
\end{center}
\caption{Binary Classification}\label{fig:Binary}
\end{figure}

\subsection{CLASSIFICATION WITH NEURAL NETWORKS}

We present an algorithm that does neural network training using the results of section \ref{sec:Control}. Each layer of the neural network is a time-step in a dynamical system, and the neural network weights correspond to the time-varying policy parameters. Let $h$ denote a component-wise nonlinearity applied to its vector-input (a transfer function). The deterministic dynamics is
\[\s_{t+1}=h\br{K_t\s_t},\s_0=\s\]
where $\s_t$ is the vector of activations at the $t$-th layer, $K_t$ is the weight matrix and $\s$ is the input to the neural network. The output is  $\s_{\Nh}$, where $\Nh$ is the number of layers in the network. The cost function is simply the loss function between the output of the neural network $\s_{\Nh}$ and a desired output $y$: $\costt\br{y,\s_{\Nh}}$. To put this into our framework, we add noise to the input of the transfer function at each layer:
 \[\s_{t+1}=h\br{K_t\s_t+\noise_t},\s_1=\s, \noise_t \sim \Gauss{0}{\Sigma_t}\]
 Additionally, we define the objective to be
 \[\ExP{}{\expb{\risk\costt\br{y,\s_{\Nh}}+\sum_{t=1}^{\Nh-1} \frac{\norm{K_t\s_t}_{\inve{\Sigma_t}}^2}{2}}}\]
where the expectation is with respect to the Gaussian noise added at each layer in the network. Note that the above objective is a function of $\Kb,\s,y$. The quadratic penalty on $K_t\s_t$ can again be thought of as a particular type of regularization which encourages learning networks with small internal activations.
We add this objective over the entire dataset $\{\s^{i},y^{i}\}$ to get our overall training objective.

We evaluate this approach on a small randomly selected subset of the MNIST dataset \citep{lecun1998gradient}. We use the version available at {\tt http://nicolas.le-roux.name/} along with the MATLAB code provided for training neural networks. We use a 2-layer neural network with 20 units in the hidden layer and tanh-transfer functions in both layers.
We use a randomly chosen collection of 900 data points for training and another 100 data points for validation. We compare training using our approach with simple backprop based training. Both of the approaches use a stochastic gradient - in our approach the stochasticity is both in selection of the data point $i$ and the realization of the Gaussian noise $\noise$ while in standard backprop the stochasticity is only in the selection of $i$. We plot learning curves (in terms of generalization or test error) for both approaches, as a function of the number of neural network evaluations (forward+back prop) performed by the algorithm in figure \ref{fig:NeuralNet}.
The nonconvex approach based on standard backprop-gradient descent gets stuck in a local minimum and does not improve test accuracy much. On the other hand, the convexified approach is able to learn a classifier that generalizes better. We also compared backprop with training a neural network on a 1-dimensional regression problem where the red curve represents the original function with data-points indicated by squares, the blue curve the reconstruction learned by our convexified training approach and the black curve the reconstruction obtained by using backprop (figure \ref{fig:NeuralNetReg}). Again, backprop gets stuck in a bad local minimum while our approach is able to find a fairly accurate reconstruction.
\begin{figure}
\begin{center}
\begin{subfigure}[b]{0.65\columnwidth}
\includegraphics[width=.98\columnwidth]{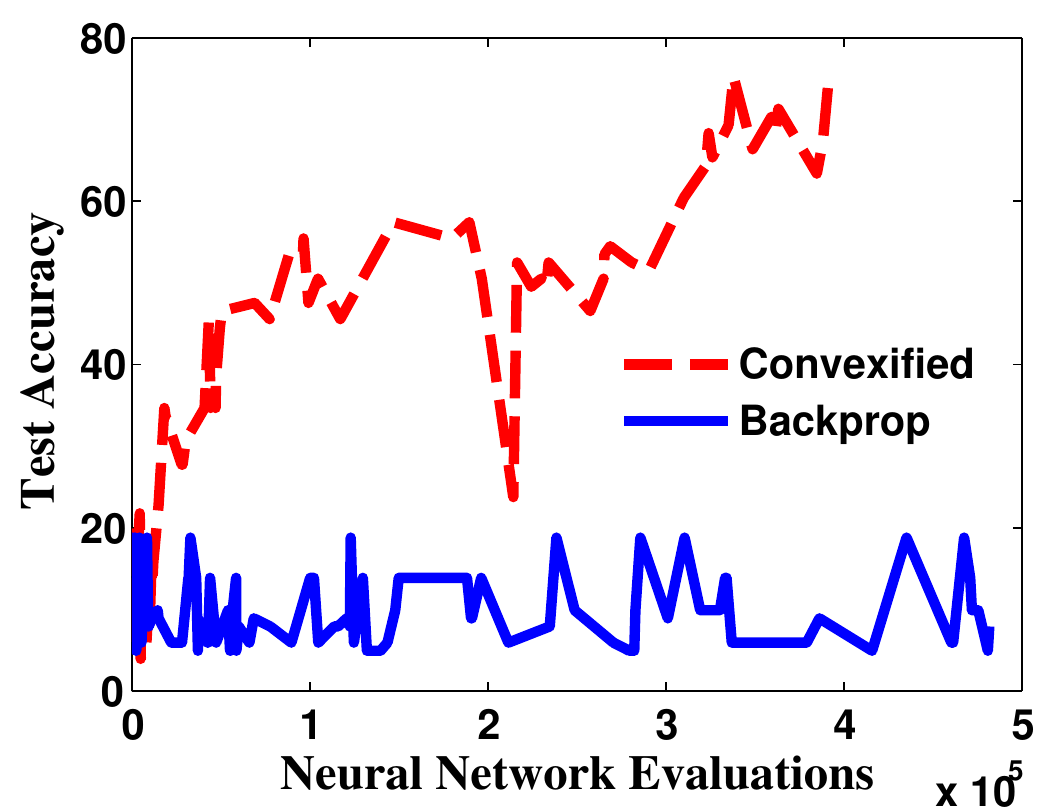}
\caption{Classification}\label{fig:NeuralNet}
 \end{subfigure}\\
 \begin{subfigure}[b]{0.6\columnwidth}
\includegraphics[width=.98\columnwidth]{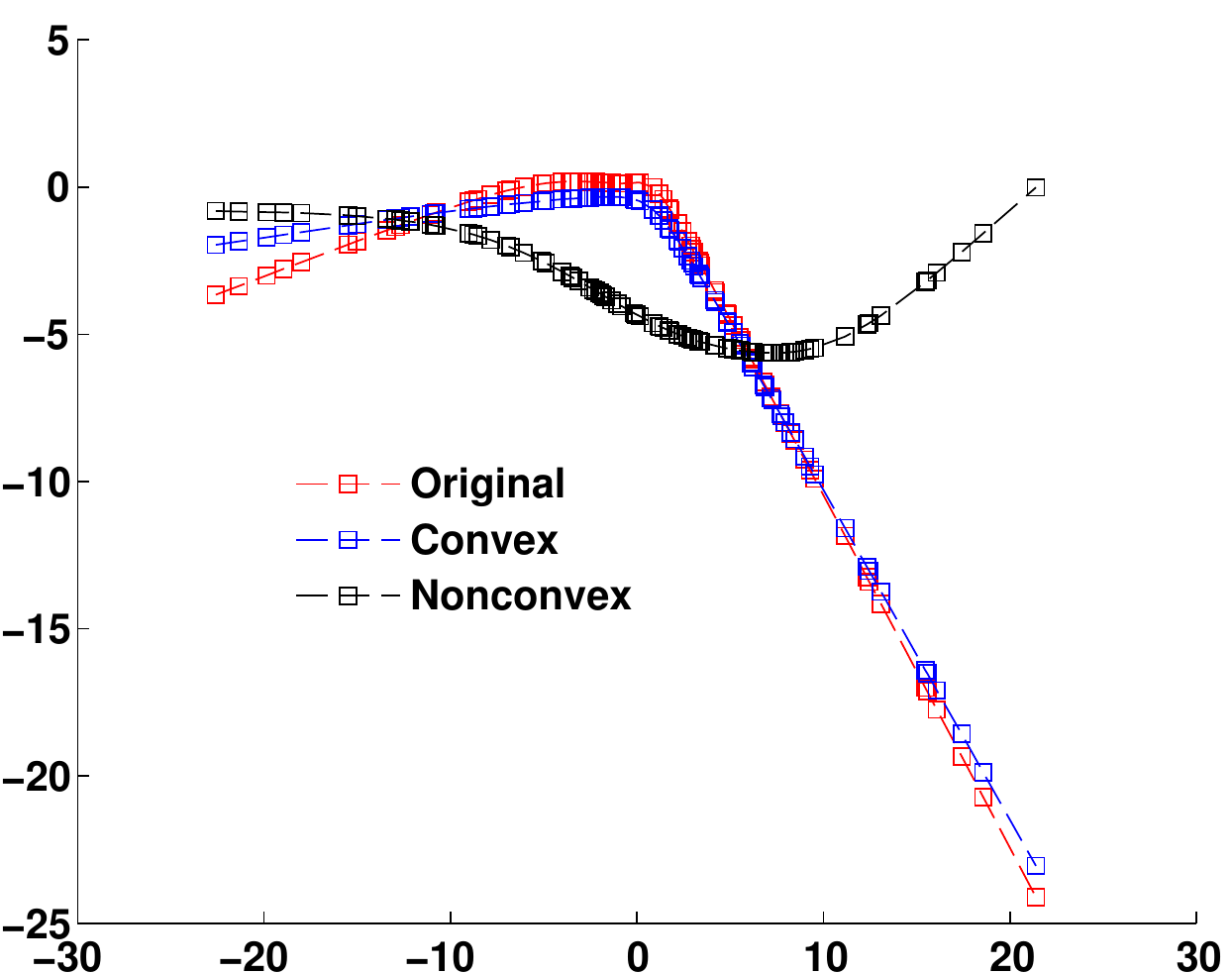}
\caption{Regression}\label{fig:NeuralNetReg}
 \end{subfigure}%
 \caption{Training Neural Networks}
\end{center}
\end{figure}

\section{CONCLUSION AND FUTURE WORK}
We have developed a general framework for convexifying a broad class of optimization problems, analysis that relates the solution of the convexified problem to the original one and given algorithms with convergence rate guarantees to solve the convexified problems. Extending the framework to dynamical systems, we derive the first approach to policy optimization with optimality and convergence rate guarantees. We validated our approach numerically on problems of binary classification and training neural networks. In future work, we will refine the suboptimality analysis for our convexification approach. Algorithmically, stochastic gradient methods could be slow if the variance in the gradient estimates is high, which is the case when using the exponentiated objective (as in section \ref{sec:Algorithms}). We will study the applicability of recent work on using better sampling algorithms with stochastic gradient \citep{StochasticProximal} to our convexified problems. 

\section*{APPENDIX}

\begin{corollary}{\citep{boucheron2013concentration}, corollary 4.15}\label{thm:KLvar}
Let $P,Q$ are arbitrary distributions on some space $\Omega$ and $f: \Omega \to \R$ is such that $\ExP{\noise\sim P}{\expb{f\br{\noise}}}<\infty$. Then, the Kullback-Leibler divergence satisfies:
\begin{align*}
& \KL{Q}{P}=  \sup_{f}\left[\ExP{\noise \sim Q}{f\br{\noise}}-\logb{\ExP{\noise \sim P}{\expb{f\br{\noise}}}}\right]
\end{align*}
\end{corollary}
\subsection{PROOF OF THEOREM \ref{thm:AlgConv}}
\newcommand{\nabf}{\tilde{\nabla} f}
\newcommand{\nabbf}{\hat{\nabla} f}
Throughout this section, Let $\tilde{\risk}=2\risk$ and $P$ denote the Gaussian density $\Gauss{0}{\Sigma}$. Define a new distribution $Q$ with density
 $Q\br{\noise} \propto P\br{\noise}\expb{\tilde{\risk} \fpertr{f}{\omega}{\theta}}$. We denote $\hat{\nabla} \Gb\br{\theta,\noise}$ by $\hat{\nabla} \Gb$, $\nabla f\br{\theta+\omega}$ by $\nabbf$ and $\nabla f\br{\theta+\omega}+\kappa\theta$ by $\nabf$ for brevity. Expectations are always with respect to $\noise\sim P$, unless denoted otherwise.
\begin{proof}
$\powb{\risk}{-2} \ExP{}{\norm{\hat{\nabla} \Gb}^2}$ evaluates to
\begin{align}
&\expb{\risk\kappa\tran{\theta}\theta} \ExP{}{\expb{\tilde\risk f\br{\noise+\theta}}\norm{\nabf}^2}=\nonumber\\
&\expb{\tilde{\risk}\br{\frac{\kappa}{2}\tran{\theta}\theta+\fpert{f}{\Sigma}{\theta}}} \ExP{}{\expb{\tilde{\risk} \fpertr{f}{\noise}{\theta}}\norm{\nabf}^2} =\nonumber\\
&  \expb{\tilde{\risk}\br{\fpert{g}{\Sigma}{\theta}-\frac{1}{2}\sigma^2\kappa}}
\ExP{}{\expb{\tilde{\risk} \fpertr{f}{\noise}{\theta}}\norm{\nabf}^2} \label{eq:Fthmexp1}
\end{align}
By the theorem hypotheses, the term outside the expectation is bounded above by $\risk^2\expb{2\risk\mup-\sigma^2\kappa}$. We are left with
\[\ExP{}{\expb{\tilde\risk \fpertr{f}{\noise}{\theta}}\norm{\nabf}^2}.\]
Dividing this by ${\ExP{}{\expb{\tilde{\risk} \fpertr{f}{\noise}{\theta}}}}$, we get
\begin{align}
\frac{\ExP{}{\expb{\tilde\risk \fpertr{f}{\noise}{\theta}}\norm{\nabf}^2}}{\ExP{}{\expb{\tilde\risk \fpertr{f}{\noise}{\theta}}}} =\ExP{\noise \sim Q}{\norm{\nabf}^2}.\label{eq:Fthmexp2}
\end{align}
Expanding $\norm{\nabf}^2$, we get
\begin{align*}
&\norm{\nabbf}^2+\kappa^2\norm{\theta}^2+2\tran{\nabbf}\theta \\
&\leq  \norm{\nabbf}^2+\kappa^2\Rad{\C}^2+2\kappa\norm{\nabbf}\Rad{\C}
\end{align*}
Finally from lemma \ref{lem:Bound}, we have
\[\ExP{\noise \sim Q}{\norm{\nabbf}^2}\leq \frac{\beta \gamma^2 }{\sigma^2\br{1-\risk\beta}}\] and by concavity of the square-root function,
\[\ExP{\noise \sim Q}{\norm{\nabbf}}\leq \sqrt{\frac{\beta \gamma^2 }{\sigma^2\br{1-\risk\beta}}}.\]
Plugging this bounds into the square expansion and letting $\delta=\sqrt{\frac{\beta \gamma^2 }{\sigma^2\br{1-\risk\beta}}}+\kappa\Rad{\C}$, we get
\begin{align}
&\ExP{\noise \sim Q}{\norm{\nabla f\br{\theta+\noise}+\kappa \theta}^2}\leq \delta^2 \label{eq:Fthmexp3}
\end{align}
From \eqref{eq:Fthmexp1},\eqref{eq:Fthmexp2} and \eqref{eq:Fthmexp3}, $\ExP{}{\norm{\hat{\nabla} \Gb}^2}$ is smaller than
\begin{align*}
 \risk^2\expb{2\risk\mup-\sigma^2\kappa}\delta^2\ExP{}{\expb{\tilde{\risk} \fpertr{f}{\omega}{\theta}}}
\end{align*}
Finally, by the last part of theorem \ref{thm:LogSobolev},
\[\ExP{}{\expb{\tilde{\risk} \fpertr{f}{\omega}{\theta}}} \leq \expb{\frac{\risk \beta}{1-\risk\beta}+2\risk\gamma^2}. \]
Combining the two above results gives the theorem.
\end{proof}
\begin{lemma}\label{lem:Bound}
Under the assumptions of theorem \ref{thm:AlgConv},
\begin{align}
\ExP{\noise \sim Q}{\sigma^2\norm{\nabla f\br{\theta+\noise}}^2} \leq \frac{\beta \gamma^2 }{1-\risk\beta}\label{eq:BoundInter}\end{align}
\end{lemma}
\begin{proof}
 The KL-divergence between $Q$ and $P$ is given by
\[\int \frac{P\br{\noise}\expb{\tilde{\risk} \fpertr{f}{\noise}{\theta}}}{\ExP{}{\expb{\tilde{\risk} \fpertr{f}{\noise}{\theta}}}}\logb{\frac{\expb{\tilde{\risk} \fpertr{f}{\noise}{\theta}}}{\ExP{}{\expb{\tilde{\risk} \fpertr{f}{\noise}{\theta}}}}}\diff \noise.\]
By theorem \ref{thm:LogSobolev} applied to $\expb{\frac{\tilde{\risk}}{2} f\br{\theta+\noise}}$), the above quantity is bounded above by $\frac{1}{2}\br{\tilde{\risk}\sigma}^2\ExP{\noise \sim Q}{\norm{\nabbf}^2}$. Then, by corollary \ref{thm:KLvar},
\begin{align*}
\KL{Q}{P} & \geq \ExP{\noise \sim Q}{\frac{\tilde{\risk}\sigma^2}{\beta}\norm{\nabbf}^2}  \nonumber \\
& \quad -\logb{\ExP{}{\expb{\frac{\tilde{\risk}\sigma^2}{\beta}\norm{\nabbf}^2}}}.
\end{align*}
Denote the second term in the RHS by $\Gamma$. Plugging in the upper bound on $\KL{Q}{P}$, we get
\[\Gamma   \geq \sigma^2\br{\frac{\tilde{\risk}}{\beta}-\frac{\tilde{\risk}^2}{2}}\ExP{\noise \sim Q}{\norm{\nabla f\br{\theta+\noise}}^2}.\]
Since the LHS is upper bounded by $\tilde{\risk} \gamma^2$ (hypothesis of theorem)
which gives us the bound
\[ \ExP{\noise \sim Q}{\sigma^2\norm{\nabla f\br{\theta+\noise}}^2}\leq \frac{2\beta \gamma^2 }{2-\tilde{\risk}\beta}=\frac{\beta \gamma^2}{1-\risk\beta}.\]
\end{proof}

\bibliography{../../../Thesis/uwthesis/bibliography}

\begin{thebibliography}{20}
\providecommand{\natexlab}[1]{#1}
\providecommand{\url}[1]{\texttt{#1}}
\expandafter\ifx\csname urlstyle\endcsname\relax
  \providecommand{\doi}[1]{doi: #1}\else
  \providecommand{\doi}{doi: \begingroup \urlstyle{rm}\Url}\fi

\bibitem[{Atchade} et~al.(2014){Atchade}, {Fort}, and
  {Moulines}]{StochasticProximal}
Y.~F. {Atchade}, G.~{Fort}, and E.~{Moulines}.
\newblock {On stochastic proximal gradient algorithms}.
\newblock \emph{ArXiv e-prints}, February 2014.

\bibitem[Bache and Lichman(2013)]{UCIDataset}
K.~Bache and M.~Lichman.
\newblock {UCI} machine learning repository, 2013.
\newblock URL \url{http://archive.ics.uci.edu/ml}.

\bibitem[Baxter and Bartlett(2001)]{Baxter2001}
J.~Baxter and P.~Bartlett.
\newblock Infinite-horizon policy-gradient estimation.
\newblock \emph{Journal of Artificial Intelligence Research}, 15:\penalty0
  319--350, 2001.

\bibitem[Boucheron et~al.(2013)Boucheron, Lugosi, and
  Massart]{boucheron2013concentration}
St{\'e}phane Boucheron, G{\'a}bor Lugosi, and Pascal Massart.
\newblock \emph{Concentration inequalities: A nonasymptotic theory of
  independence}.
\newblock Oxford University Press, 2013.

\bibitem[Boyd and Vandenberghe(2004)]{Boyd04}
S.~Boyd and L.~Vandenberghe.
\newblock \emph{Convex optimization}.
\newblock Cambridge University Press, Cambrdige, UK, 2004.

\bibitem[Broek et~al.(2010)Broek, Wiegerinck, and Kappen]{KappenRisk}
B.~Van~Den Broek, W.~Wiegerinck, and B.~Kappen.
\newblock Risk sensitive path integral control.
\newblock In \emph{Uncertainty in AI, 2010. Proceedings of the 2010}, 2010.

\bibitem[Bubeck(2013)]{BubeckNotes}
Sebastian Bubeck.
\newblock The complexities of optimization, December 2013.
\newblock URL
  \url{https://blogs.princeton.edu/imabandit/2013/04/25/orf523-noisy-oracles/}.

\bibitem[Chang and Lin(2011)]{CC01a}
Chih-Chung Chang and Chih-Jen Lin.
\newblock {LIBSVM}: A library for support vector machines.
\newblock \emph{ACM Transactions on Intelligent Systems and Technology},
  2:\penalty0 27:1--27:27, 2011.
\newblock Software available at \url{http://www.csie.ntu.edu.tw/~cjlin/libsvm}.

\bibitem[Dvijotham and Todorov(2011)]{Dvijotham2011}
K.~Dvijotham and E.~Todorov.
\newblock A unifying framework for linearly-solvable control.
\newblock \emph{Uncertainty in Artificial Intelligence}, 2011.

\bibitem[Fleming and Mitter(1982)]{Fleming82}
W.~Fleming and S.~Mitter.
\newblock Optimal control and nonlinear filtering for nondegenerate diffusion
  processes.
\newblock \emph{Stochastics}, 8:\penalty0 226--261, 1982.

\bibitem[Kappen(2005)]{Kappen05}
H.J. Kappen.
\newblock {Linear theory for control of nonlinear stochastic systems}.
\newblock \emph{Physical Review Letters}, 95\penalty0 (20):\penalty0 200201,
  2005.

\bibitem[LeCun et~al.(1998)LeCun, Bottou, Bengio, and
  Haffner]{lecun1998gradient}
Yann LeCun, L{\'e}on Bottou, Yoshua Bengio, and Patrick Haffner.
\newblock Gradient-based learning applied to document recognition.
\newblock \emph{Proceedings of the IEEE}, 86\penalty0 (11):\penalty0
  2278--2324, 1998.

\bibitem[Marcus et~al.(1997)Marcus, Fern{\'a}ndez-Gaucherand,
  Hern{\'a}ndez-Hernandez, Coraluppi, and Fard]{marcus1997risk}
S.I. Marcus, E.~Fern{\'a}ndez-Gaucherand, D.~Hern{\'a}ndez-Hernandez,
  S.~Coraluppi, and P.~Fard.
\newblock {Risk sensitive Markov decision processes}.
\newblock \emph{Systems and Control in the Twenty-First Century}, 29, 1997.

\bibitem[Mobahi(2012)]{mobahi2012phd}
Hossein Mobahi.
\newblock \emph{Optimization by Gaussian Smoothing with Application to
  Geometric Alignment}.
\newblock PhD thesis, University of Illinois at Urbana Champaign, Dec 2012.

\bibitem[Rangarajan(1990)]{Rangarajan90generalisedgraduated}
A.~Rangarajan.
\newblock Generalised graduated non-convexity algorithm for maximum a
  posterjori image estimation.
\newblock In \emph{Proc. ICPR}, pages 127--133, 1990.

\bibitem[Speyer et~al.(1974)Speyer, Deyst, and
  Jacobson]{speyer1974optimization}
J~Speyer, John Deyst, and D~Jacobson.
\newblock Optimization of stochastic linear systems with additive measurement
  and process noise using exponential performance criteria.
\newblock \emph{Automatic Control, IEEE Transactions on}, 19\penalty0
  (4):\penalty0 358--366, 1974.

\bibitem[Theodorou et~al.(2010{\natexlab{a}})Theodorou, Buchli, and
  Schaal]{Theodorou2010}
E.~Theodorou, J.~Buchli, and S.~Schaal.
\newblock {Reinforcement learning of motor skills in high dimensions: A path
  integral approach}.
\newblock In \emph{Robotics and Automation (ICRA), 2010 IEEE International
  Conference on}, pages 2397--2403. IEEE, 2010{\natexlab{a}}.

\bibitem[Theodorou et~al.(2010{\natexlab{b}})Theodorou, Buchli, and
  Schaal]{theodorou2010generalized}
Evangelos Theodorou, Jonas Buchli, and Stefan Schaal.
\newblock A generalized path integral control approach to reinforcement
  learning.
\newblock \emph{The Journal of Machine Learning Research}, 9999:\penalty0
  3137--3181, 2010{\natexlab{b}}.

\bibitem[Todorov(2009)]{Todorov09}
E.~Todorov.
\newblock {Efficient computation of optimal actions}.
\newblock \emph{Proceedings of the National Academy of Sciences}, 106\penalty0
  (28):\penalty0 11478, 2009.

\bibitem[Toussaint(2009)]{Toussaint09}
M.~Toussaint.
\newblock Robot trajectory optimization using approximate inference.
\newblock \emph{International Conference on Machine Learning}, 26:\penalty0
  1049--1056, 2009.

\end{thebibliography}
\bibliographystyle{plainnat}

\end{document}